\documentclass[11pt,a4paper,DIV20,headsepline,abstracton]{scrartcl}
\pdfoutput=1

\usepackage{cmap}
\usepackage{ucs}
\usepackage[utf8x]{inputenc}
\usepackage[greek,english]{babel}

\usepackage{amsmath,amssymb,amsthm}
\allowdisplaybreaks
\usepackage{graphics,graphicx}
\usepackage{wrapfig}
\usepackage[margin=1.15in]{geometry}
\usepackage{dsfont}
\usepackage{enumerate}
\usepackage{mathrsfs}
\usepackage{mathtools}
\usepackage{paralist}
\usepackage{color}
\usepackage{transparent}
\usepackage[hidelinks]{hyperref}
\usepackage[affil-it]{authblk}

\usepackage{svg}

\usepackage{scrpage2}

\usepackage{microtype}

\usepackage{booktabs}

\usepackage[numbers]{natbib}
\usepackage{chapterbib}

\usepackage{url}

\usepackage{xfrac}

\usepackage{pgfplots}

\usepackage[titletoc,title]{appendix}

\usepackage{epigraph}

\usepackage[font=footnotesize]{caption}

\numberwithin{equation}{section}

\numberwithin{equation}{section}

\theoremstyle{plain}
\newtheorem{thm}{Theorem}[section]
\newtheorem{propn}[thm]{Proposition}
\newtheorem{lemma}[thm]{Lemma}
\newtheorem{cor}[thm]{Corollary}

\newtheorem{conj}[thm]{Conjecture}

\theoremstyle{definition}
\newtheorem{defn}[thm]{Definition}

\theoremstyle{remark}
\newtheorem{rk}[thm]{Remark}

\newcommand{\RR}{\mathbb{R}}

\newcommand{\ZZ}{\mathbb{Z}}

\newcommand{\eps}{\varepsilon}

\newcommand{\dd}{\mathrm{d}}

\newcommand{\UU}{\mathcal U}

\newcommand{\Qq}{\mathcal Q}

\newcommand{\MM}{\mathcal M}
\newcommand{\HH}{\mathcal H}
\newcommand{\Nn}{\mathcal N}
\newcommand{\Rr}{\mathcal R}
\newcommand{\II}{\mathcal I}
\newcommand{\IIa}{\II_{\mathrm{aux}}}
\newcommand{\LL}{\mathcal L}
\newcommand{\Ss}{\mathcal S}

\newcommand{\Tt}{\mathcal T}

\newcommand{\OO}{\mathcal O}

\newcommand{\del}{\partial}

\renewcommand{\epsilon}{\varepsilon}

\renewcommand{\phi}{\varphi}
\renewcommand{\theta}{\vartheta}
\renewcommand{\rho}{\varrho}

\newcommand*\Laplace{\mathop{}\!\mathbin\bigtriangleup}
\newcommand{\SLaplace}{\Laplace \kern-3mm /\kern+1mm}
\newcommand{\SGrad}{\nabla \kern-2.9mm /\kern+1mm}

\newcommand{\rrr}{\tilde r}
\newcommand{\ovt}[1]{\overline{\tilde{#1}}}
\newcommand{\ov}[1]{\overline{#1}}

\newcommand{\e}{\operatorname{e}}

\newcommand{\spur}{\mathrm{tr}}

\author{Dominic Dold}
\title{Global dynamics\\ of asymptotically locally AdS spacetimes\\ with negative mass}
\date{}
\author{Dominic Dold\footnote{\,\,\,Cambridge Centre for Analysis, Department of Pure Mathematics and Mathematical Statistics, University of Cambridge, Wilberforce Road, Cambridge CB3 0AG, United Kingdom\\Email address: %\tt{\href{mailto:D.Dold@maths.cam.ac.uk}{D.Dold@maths.cam.ac.uk}}}}
	\href{mailto:D.Dold@maths.cam.ac.uk}{{D.Dold@maths.cam.ac.uk}}}}
\date{}

\begin{document}

\maketitle
\vspace{-1cm}
\begin{abstract}
The Einstein vacuum equations in five dimensions with negative cosmological constant are studied in biaxial Bianchi IX symmetry. We show that if initial data of Eguchi-Hanson type, modelled after the four-dimensional Riemannian Eguchi-Hanson space, have negative mass, the future maximal development does not contain horizons, i.\,e. the complement of the causal past of null infinity is empty. In particular, perturbations of Eguchi-Hanson-AdS spacetimes within the biaxial Bianchi IX symmetry class cannot form horizons, suggesting that such spacetimes are potential candidates for a naked singularity  to form.
The proof  relies on an extension principle proven for this system and a priori estimates following from the monotonicity of the Hawking mass.
\end{abstract}
\vspace{.2cm}
\tableofcontents

\section{Introduction}

\subsection{The Einstein vacuum equations with negative cosmological constant}

The Einstein vacuum equations in $n$ dimensions ($n>2$)
\begin{align}
\label{eqn:EVE_general_dimension}
\mathrm{Ric}(g)=\frac{2}{n-2}\Lambda g
\end{align}
with cosmological constant $\Lambda\in \RR$ can be understood as a system of second-order partial differential equations for the metric tensor $g$ of an $n$-dimensional spacetime $\left(\MM,g\right)$. Solutions with negative cosmological constant $\Lambda=-(n-1)(n-2)/(2\ell^2)<0$  have drawn considerable attention in recent years, mainly due to the conjectured instability of these spacetimes. For more details, see \citep{Anderson_uniqueness}, \citep{Dafermos-Holzegel_EH}, \citep{Bizon}, \citep{DiasGravitational}, \citep{DiasHorowitzMarolfSantos}, \citep{HolzegelLukSmuleviciWarnick} and references therein.

The system (\ref{eqn:EVE_general_dimension}) is of hyperbolic nature, and studying the dynamic evolution of initial data  is very difficult in general, leading us to take recourse to settings with high degrees of symmetry. In particular, it is desirable to reduce the dimension of the dynamical problem to the simplest case of $1+1$ dimensions. This approach has a longer history for $\Lambda=0$. There, for $n=4$, the only symmetry group achieving the reduction to a $1+1$-dimensional problem whilst consistent with the spacetime being asymptotically flat  is spherical symmetry. However, the well-known Birkhoff theorem prevents any dynamical consideration since such a four-dimensional spacetime is necessarily static, embedding locally into a subset of a member of the Schwarzschild family. 

To study spherically symmetric gravitational dynamics in four dimensions, one can follow the approach of the seminal work by Christodoulou and couple gravity to matter. In a sequence of papers -- see his own survey article \citep{Christodoulou_survey} for references -- he initiated the rigorous analysis of spherically symmetric gravitational collapse for $\Lambda=0$ by studying the Einstein-scalar field system. The model of a real massless scalar field was chosen because, on the one hand, this matter model does not develop singularities in the absence of gravity and, on the other hand, its wave-like character resembles the character of general gravitational perturbations of Minkowski space. Christodoulou's work led to a complete understanding of weak and strong cosmic censorship for this model. His approach has later been extended to other matter models; see \citep{Kommemi} for a systematic overview and references.

Christodoulou's approach was adapted to the context of $\Lambda<0$ by Holzegel and Smulevici in \citep{Holzegel-Smulevici_self-gravitating} and by Holzegel and Warnick in \citep{Holzegel-Warnick_Einstein-Klein-Gordon}, who show well-posedness of the Einstein-Klein-Gordon system with the scalar field satisfying various reflecting boundary conditions at infinity. The work \citep{Holzegel_Smulevici_stability} shows stability of Schwarzschild-AdS in this symmetry class for Dirichlet boundary conditions. A recent breakthrough has been achieved by Moschidis in   \citep{Moschidis_instability} and \citep{Moschidis_well-posedness}; in his work, he shows instability of exact anti-de\,Sitter space as a solution to the Einstein-null dust system in spherical symmetry with an inner mirror.

Another possibility of evading the restrictions of Birkhoff's Theorem is to study (\ref{eqn:EVE_general_dimension}) in higher dimensions. Working in five dimensions and imposing biaxial Bianchi IX symmetry, a symmetry corresponding to a subgroup of $SO(4)$, still reduces the system to 1+1 dimensions and introduces a dynamical variable $B$, not dissimilar to the scalar field in the coupled system. 
This model was introduced by Biz\'on, Chmaj and Schmidt. In \citep{Bizon-Chmaj-Schmidt}, they initiated the study of gravitational collapse for $\Lambda=0$ in this symmetry class by numerical computations; investigations along those lines were continued in \citep{Dafermos-Holzegel_triaxial} and \citep{Bizon-Chamaj-Schmidt_triaxial}.

The study of this system in the realm of negative $\Lambda$ has been initiated by Dafermos and Holzegel in 2006. In \citep{Dafermos-Holzegel_EH} -- now mostly cited for the conjecture of the instability of exact AdS space --, our Corollary~\ref{cor:no_horizons_perturbations} has been put forward without rigorous proof. Back then, the problem of proving local well-posedness for the system in biaxial Bianchi IX symmetry was not solved, thus no extension principle sufficiently strong was available. 
The present paper can be seen as a completion of \citep{Dafermos-Holzegel_EH}, building on the insight into problems in asymptotically locally AdS spacetimes obtained over the past decade;  for an overview of this work, see \citep{Holzegel-Warnick_Einstein-Klein-Gordon}, \citep{Enciso-Kamran} and references therein. 

In five dimensions and for $\Lambda<0$, (\ref{eqn:EVE_general_dimension}) has many static solutions which are asymptotically locally AdS.  A spacetime is asymptotically locally AdS if the  asymptotics of the metric towards conformal infinity $\II$  is modelled after AdS space, but $\II$ need not be $\RR\times S^3$ topologically. Prominent examples of such static solutions are exact $\mathrm{AdS}_5$ space with spherical conformal infinity\footnote{\,\,\,Numerical studies within the biaxial Bianchi IX symmetry class for perturbations of $\mathrm{AdS}_5$ were carried out recently in \citep{Bizon-Rostworowski_AdS5}.}  and the AdS soliton of \citep{Horowitz-Myers_positive} with toric $\II$. Eguchi-Hanson-AdS spacetimes form another such family with $\II\cong \RR\times (S^3/\ZZ_n)$ for $n\geq 3$.\footnote{\,\,\,The space $S^3/\ZZ_n$ is defined in the usual way as the lense space $L(n,1)$.}

\subsection{Spaces of Eguchi-Hanson type}

We introduce four-dimensional Riemannian manifolds modifying Eguchi-Hanson space to the asymptotically locally AdS context. These will serve as initial data for the five-dimensional Einstein vacuum equations
\begin{align}
\label{eqn:EVE_five}
\mathrm{Ric}(g)=\frac{2}{3}\Lambda g
\end{align}
via the local well-posedness Theorem~\ref{thm:general_local_existence}. Our data also exhibit an $SU(2)\times U(1)$ symmetry, thus giving rise to spacetimes with  biaxial Bianchi IX symmetry. Then Eguchi-Hanson-AdS spacetimes form particular examples of the spacetimes thus obtained.

\begin{defn}
	\label{defn:Eguchi-Hanson-type}
	We say that  an initial data set $(\Ss,\ov{g},K)$ to (\ref{eqn:EVE_five}) exhibiting $SU(2)\times U(1)$ symmetry is of Eguchi-Hanson type if
	\begin{align*}
	\Ss=(a,\infty)\times\left(S^3/\ZZ_n\right)
	\end{align*}
	for fixed $a>0$ and $n\geq 3$ satisfying
	\begin{align*}
	\frac{n^2}{4}=1+\frac{a^2}{\ell^2},
	\end{align*}
	and if 
	\begin{align*}
	\ov{g}=\frac{1}{A}\dd\rho^2+\gamma~~~~\mathrm{with}~~~~
	\gamma=\frac{1}{4}r^2\e^{2B}\left(\sigma_1^2+\sigma_2^2\right)+\frac{1}{4}r^2\e^{-4B}\,\sigma_3^2,
	\end{align*}
	where $(\sigma_1,\sigma_2,\sigma_3)$ is a basis of left-invariant one-forms on $SU(2)$ (see below) and 	\begin{align*}
	A,\,r,\,B:\,(a,\infty)\rightarrow\RR
	\end{align*}
	are smooth functions such that the following conditions hold:
	\begin{compactenum}
		\item[(i)] Around the centre $\rho=a$, the functions satisfy the regularity conditions
		\begin{align*}
		A=&\frac{n^2}{a}(\rho-a)+\OO\left((\rho-a)^2\right)\\
		r=&2^{1/3}a^{5/6}(\rho-a)^{1/6}+\OO\left((\rho-a)^{7/6}\right)\\
		B+\log r=&\log a+\OO((\rho-a)).
		\end{align*}
		\item[(ii)] The function $A$ is non-zero on $(a,\infty)$. Moreover 
		\begin{align*}
		A=\frac{\rho^2}{\ell^2}+1+o(1)
		\end{align*}
		as $\rho\rightarrow\infty$.
		\item[(iii)] The function $r$ is the radius of the 3-spheres at $\rho$. Moreover
		\begin{align*}
		r=\rho+\OO(1)
		\end{align*}
		as $\rho\rightarrow\infty$.
		\item[(iv)] For an $R>a$, we have
		\begin{align*}
		\int_{R}^{\infty}\left(\rho^3B^2+\rho^7\left(\del_{\rho} B\right)^2\right)\,\dd\rho <C~~~\mathrm{and}~~~
		\sup_{\rho\in(R,\infty)}|\rho^3\del_{\rho}B|<C.
		\end{align*}
	\end{compactenum}
	We require that
	\begin{align*}
	M:=\lim_{\rho\rightarrow\infty}\left(\frac{r^2}{2}\left[1+\frac{r^2}{12}\left(\left(\spur_{\gamma}K\right)^2-H^2\right)\right]+\frac{r^4}{2\ell^2}\right),
	\end{align*}
	where $H$ is the mean curvature of the symmetry orbits, is finite; we call $M$ the mass of $(\Ss,\ov{g},K)$ at infinity.
\end{defn}

\begin{rk}
\label{rk:after_EH-type}
	\begin{compactenum}
	\item 	The left-invariant one-forms satisfy
		\begin{align*}
		\dd\sigma_1+\sigma_2\land\sigma_3=0,~\dd\sigma_2+\sigma_3\land\sigma_1=0,~\dd\sigma_3+\sigma_1\land\sigma_2=0.
		\end{align*}
		One can choose Euler angles $(\theta,\phi,\psi)$, $0<\theta<\pi$, $0\leq\phi< 2\pi$, $0\leq\psi< 4\pi$ on $SU(2)$ such that
		\begin{align*}
		\sigma_1=\sin\theta\,\sin\psi\,\dd\phi+\cos\psi\,\dd\theta,~~\sigma_2=\sin\theta\,\cos\psi\,\dd\phi-\sin\psi\,\dd\theta,~~
		\sigma_3=&\cos\theta\,\dd\phi+\dd\psi.
		\end{align*}
		In terms of the left-invariant one-forms, the Minkowski metric on $\RR^{5}$ is given by
		\begin{align*}
		g_{\mathrm{Mink}}=-\dd t^2+\dd r^2+\frac{1}{4}r^2\left(\sigma_1^2+\sigma_2^2+\sigma_3^2\right).
		\end{align*}
		The Euler angles $(\theta,\phi,\psi)$ parametrise the 3-sphere away from the poles. By restricting $\psi$ to have period $4\pi/n$, we obtain coordinates on $S^3/\ZZ_n$.
	\item By identity (\ref{eqn:delr_hypersurface}), the notion of mass at infinity is consistent with the renormalised Hawking mass introduced in Definition~\ref{defn:biaxial}.
	\end{compactenum}
\end{rk}

Prima facie it seems as if $\ov{g}$ had a singularity at $\rho=a$. However, one should compare this situation to that of spherical symmetry in spherical coordinates. This intuition is made more precise in the following

\begin{propn}
	\label{propn:centre_EH_space}
	Let $(\Ss,\ov{g})$ be of Eguchi-Hanson type. Then there is a $b>0$ such that for $\rho\in(a,b)$, $(\Ss\cap\{\rho\leq b\},\ov{g}\lvert_{\{\rho\leq b\}})$ has topology $\RR^2\times S^2$ and can be smoothly compactified by adding a 2-sphere at $\rho=a$. The resulting manifold is smooth and has no boundary.
\end{propn}

\begin{proof}
	Define
	\begin{align*}
	z:=\frac{4\sqrt{a}}{n}\left(\rho-a\right)^{1/2}.
	\end{align*}
	To leading order, we have
	\begin{align*}
	r^2\e^{-4B}\sim\frac{r^6}{a^4}\sim 4a\left(\rho-a\right)=\frac{1}{4}n^2z^2
	\end{align*}
	around $\rho=a$. Thus the metric becomes
	\begin{align*}
	\ov{g}\sim \frac{1}{4}\left(\dd z^2+n^2z^2\left(\dd\psi+\cos\theta\dd\phi\right)^2\right)+\frac{\rho^2}{4}\left(\dd\theta^2+\sin^2\theta\dd\phi^2\right)
	\end{align*}
	to leading order. For fixed $(\theta,\phi)$, the restriction on the range of $\psi$  (see Remark~\ref{rk:after_EH-type}) guarantees that the metric can be extended smoothly to $\rho=a$. By adding an $S^2$ at $\rho=a$, we obtain a manifold without boundary that has local topology $\RR^2\times S^2$.
\end{proof}

\begin{rk}
We immediately see that at infinity, $(\Ss,\ov{g})$ is asymptotically locally AdS.
\end{rk}

%\begin{wrapfigure}{r}{.55\textwidth}
\begin{figure}
	\begin{center}
	\vspace{-1.2cm}
		%\hspace{.2cm}
		\def\svgwidth{250pt}
		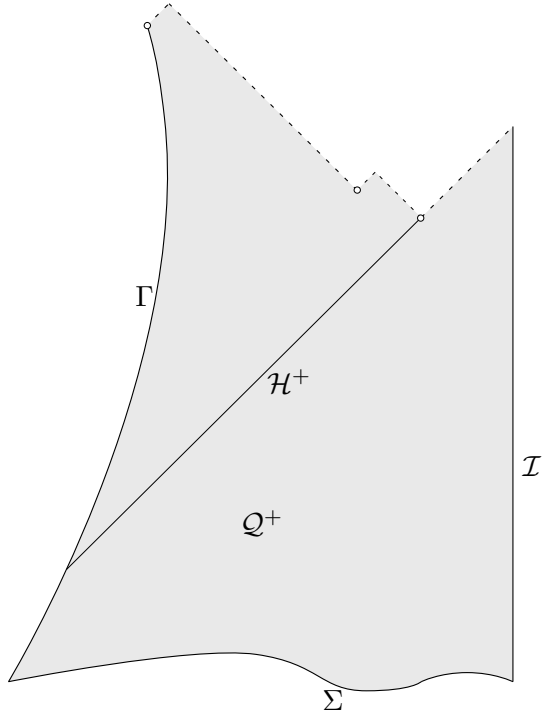
		\caption{A general Penrose diagram of the future maximal development of initial data of Eguchi-Hanson type via Theorem~\ref{thm:general_local_existence}\label{fig:general_Penrose}}
	\end{center}
\end{figure}
%\end{wrapfigure}

Given Eguchi-Hanson-type initial data, the Einstein vacuum equations \eqref{eqn:EVE_five} are well-posed in the biaxial Bianchi IX symmetry class -- see \citep{Dold_thesis} for a proof. We merely state the well-posedness theorem here.

\begin{defn}
	\label{defn:biaxial}
	Let $(\MM,g)$ be a five-dimensional spacetime. Then $(\MM,g)$ exhibits a biaxial Bianchi IX symmetry if, topologically,
	\begin{align*}
	\MM=\Qq\times \left(S^3/\Gamma\right)
	\end{align*}
	for $\Qq$ a two-dimensional manifold (possibly with boundary) and a discrete group $\Gamma\in\{\emptyset,\ZZ_2,\ZZ_3,\ldots\}$ such that
	\begin{align}
	\label{eqn:metric_Bianchi}
	g=&h+\frac{1}{4}r^2\left(\e^{2B}(\sigma_1^2+\sigma_2^2)+\e^{-4B}\sigma_3^2\right).
	\end{align}
	Here $h$ is a Lorentzian metric on $\Qq$, and $r$ and $B$ are smooth real-valued functions on $\Qq$. The value $r(q)$ is the area radius of the squashed sphere through $q\in\Qq$, i.\,e.
	\begin{align*}
	2\pi^2r^3={\mathrm{vol}\left(S^3_q\right)},
	\end{align*}
	where $S^3_q$ is the sphere at $q$. In this symmetry class, we introduce the renormalised Hawking mass   (henceforth referred to as Hawking mass)
	\begin{align*}
	m:\,\Qq\rightarrow\RR
	\end{align*}
	 by
	\begin{align*}
	m=\frac{r^2}{2}\left(1-g\left(\nabla r,\nabla r\right)\right)+\frac{r^4}{2\ell^2}.
	\end{align*}
\end{defn}

\begin{defn}
	A spacetime $(\MM,g)$ exhibiting biaxial Bianchi IX symmetry is asymptotically locally AdS with radius $\ell$ and conformal infinity $\II$ if it is conformally equivalent to a manifold $(\tilde{\MM},\tilde g)$ with boundary $\II:=\del\tilde\MM$ such that 
	\begin{compactenum}[(i)]
	\item Conformal infinity $\II$ has topology $\RR\times (S^3/\Gamma)$.
	\item The inverse $\rrr:=r^{-1}$ is a boundary defining function for $\II$, i.\,e. $\rrr=0$ and $\dd\rrr\neq 0$ on $\II$.
	\item The rescaled metric $r^{-2}g$ is a smooth metric on a neighbourhood of $\II$ in $\tilde{\MM}$.
	\item For small $\rrr>0$, there exist coordinates $(t,\rrr)$ on $\Qq$ such that, locally,	
	\begin{align*}
	h\left(\frac{1}{\rrr^2}\del_{\rrr},\frac{1}{\rrr^2}\del_{\rrr}\right)={\ell^2}\rrr^2+\OO\left(\rrr^4\right).
	\end{align*}
	in a neighbourhood of $\II$.
	\item The quantity $B$ satisfies a Dirichlet boundary condition, i.\,e. $B=0$ on $\II$.
	\end{compactenum}
\end{defn}

\begin{thm}
	\label{thm:general_local_existence}
	Let $(\Ss,\ov{g},K)$ be an asymptotically locally AdS initial data set with mass $M$ at infinity such that $(\Ss,\ov{g},K)$ is of Eguchi-Hanson type. Then there is a $T>0$ and a manifold $\MM:=(-T,T)\times\Ss$ equipped with a metric $g$ exhibiting biaxial Bianchi IX symmetry such that $(\MM,g)$ is asymptoticall locally AdS, $g$ solves \eqref{eqn:EVE_five} and $\{0\}\times \Ss$ has induced metric $\ov{g}$ and second fundamental form $K$. Moreover, $(\MM,g)$ is the unique asymptotically locally AdS solution to \eqref{eqn:EVE_five} with initial data $(\Ss,\ov{g},K)$.
\end{thm}

\begin{rk}
	\label{rk:maximal_development}
	\begin{compactenum}
	\item The local well-posedness theorem for an initial data set $(\Ss,\ov{g},K)$ yields the existence of a unique maximal development in the sense of \citep{Holzegel-Smulevici_self-gravitating}.
	\item The proof of the local well-posedness theorem in \citep{Dold_thesis} proceeds along the lines of \citep{Holzegel-Warnick_Einstein-Klein-Gordon}. Well-posedness of the Einstein vacuum equations for $\Lambda<0$ in four dimensions without symmetry assumptions was shown by Friedrich in  \citep{Friedrich}, and a recent generalisation to higher dimensions by Enciso and Kamran is also available; see \citep{Enciso-Kamran}. In particular, Theorem~\ref{thm:general_local_existence} follows from their work. However the theorem as stated is too general for an extension principle (Section~\ref{sec:extension_principle}); so to exploit the monotonicity of the Hawking mass (see Section~\ref{sec:a_priori_no_horizon}), a local well-posedness result in norms propagated by the mass (as in Section~\ref{sec:local_well}) is required.
	\end{compactenum}
\end{rk}

The explicit examples behind this well-posedness theorem are Eguchi-Hanson-AdS spacetimes, constructed in  \citep{Clarkson-Mann}. They form a family of solutions  $(\MM_{\mathrm{EH},a},g_{\mathrm{EH},a})$ to (\ref{eqn:EVE_general_dimension}) in five dimensions. For fixed $\Lambda=-6/\ell^2<0$, they form a one-parameter family of static spacetimes exhibiting biaxial Bianchi IX symmetry.
If we define
\begin{align*}
r=\rho\left(1-\frac{a^4}{\rho^4}\right)^{1/6},~~~~
\Omega^2=1+\frac{\rho^2}{\ell^2},~~~~
B=-\frac{1}{6}\log\left(1-\frac{a^4}{\rho^4}\right)
\end{align*}
and choose coordinates such that
\begin{align*}
h=-\frac{1}{2}\Omega^2\left(\dd u\otimes\dd v+\dd v\otimes\dd u\right),
\end{align*}
with 
\begin{align*}
\dd u=\dd t-\frac{1}{\left(1+\frac{\rho^2}{\ell^2}\right)\left(1-\frac{a^4}{\rho^4}\right)^{1/2}}\dd \rho,~~~~
\dd v=\dd t+\frac{1}{\left(1+\frac{\rho^2}{\ell^2}\right)\left(1-\frac{a^4}{\rho^4}\right)^{1/2}}\dd \rho,
\end{align*}
the metric takes the form
\begin{align*}
g_{\mathrm{EH},a}=&-\left(1+\frac{\rho^2}{\ell^2}\right)\,\dd t^2+\frac{1}{\left(1+\frac{\rho^2}{\ell^2}\right)\left(1-\frac{a^4}{\rho^4}\right)}\,\dd\rho^2\\&+\frac{1}{4}\left(1-\frac{a^4}{\rho^4}\right)\left(\dd\psi+\cos\theta\,\dd\phi\right)^2+\frac{\rho^4}{4}\left(\dd\theta^2+\sin^2\theta\,\dd\phi^2\right)
\end{align*}
in $(t,\rho,\theta,\phi,\psi)$ variables with $\rho\in(a,\infty)$. In the limit $\ell\rightarrow\infty$, the metric $g_{\mathrm{EH},a}$ restricted to hypersurfaces of constant $t$ yields the Riemannian Eguchi-Hanson metric, which was first presented in \citep{Eguchi-Hanson}.

We immediately note:

\begin{propn}
	Let $(\MM_{\mathrm{EH},a},g_{\mathrm{EH},a})$ be an Eguchi-Hanson-AdS spacetime. Then
	\begin{align*}
	M_{\mathrm{EH},a}:=\lim_{\rho\rightarrow\infty}m=-\frac{1}{6}\frac{a^4}{\ell^2}
	\end{align*}
	is negative. At the centre $\rho=a$, the Hawking mass is ill-defined, tending to $-\infty$.
\end{propn}

\subsection{The main result}
\label{sec:main_results}

The main novel result of this paper consists in showing that for initial data of Eguchi-Hanson type with negative mass, a Penrose diagram such as Figure~\ref{fig:global_geometry} cannot arise.

\begin{thm}
	\label{thm:no_horizons_intro}
	Let $(\Ss,\ov{g},K)$ be of Eguchi-Hanson type with negative mass $M<0$ at infinity. Then there is no future horizon in the maximal development, i.\,e. the causal past of null infinity is empty.
\end{thm}

\begin{cor}
	\label{cor:no_horizons_perturbations}
	Small perturbations of Eguchi-Hanson-AdS spacetimes do not contain future horizons.
\end{cor}

It is important to stress that the absence of a horizon is a stronger statement than the absence of trapping -- shown in Proposition~\ref{propn:no_trapping} -- for a horizon concerns the causal past of null infinity and hence the global geometry of the spacetime, whereas trapping is a local phenomenon.

Combining Theorem~\ref{thm:no_horizons_intro} with the arguments in Section~\ref{sec:global_biaxial_Bianchi_IX} leaves us with the following dichotomy: either the future development of Eguchi-Hanson-type data with negative mass contains a first singularity in $\ov{\Gamma}\backslash\Gamma$, where $\Gamma$ is the centre -- see Figure~\ref{fig:naked_singularity} --, or no first singularities form at all. 

In virtue of the properties and conjectures described in the next section, our result, restricted to perturbations of Eguchi-Hanson-AdS spacetimes, can corroborate the conjecture put forward in in \citep{Dafermos-Holzegel_EH}:

\begin{conj}
	Small perturbations of Eguchi-Hanson-AdS spacetimes have a Penrose diagram as depicted in Figure~\ref{fig:naked_singularity}. Moreover, $\II$ is future incomplete.
\end{conj}

%\begin{wrapfigure}{r}{.45\textwidth}
\begin{figure}
	\begin{minipage}{.45\textwidth}
	\begin{center}
		\vspace{-1cm}
		%\hspace{.2cm}
		\def\svgwidth{220pt}
		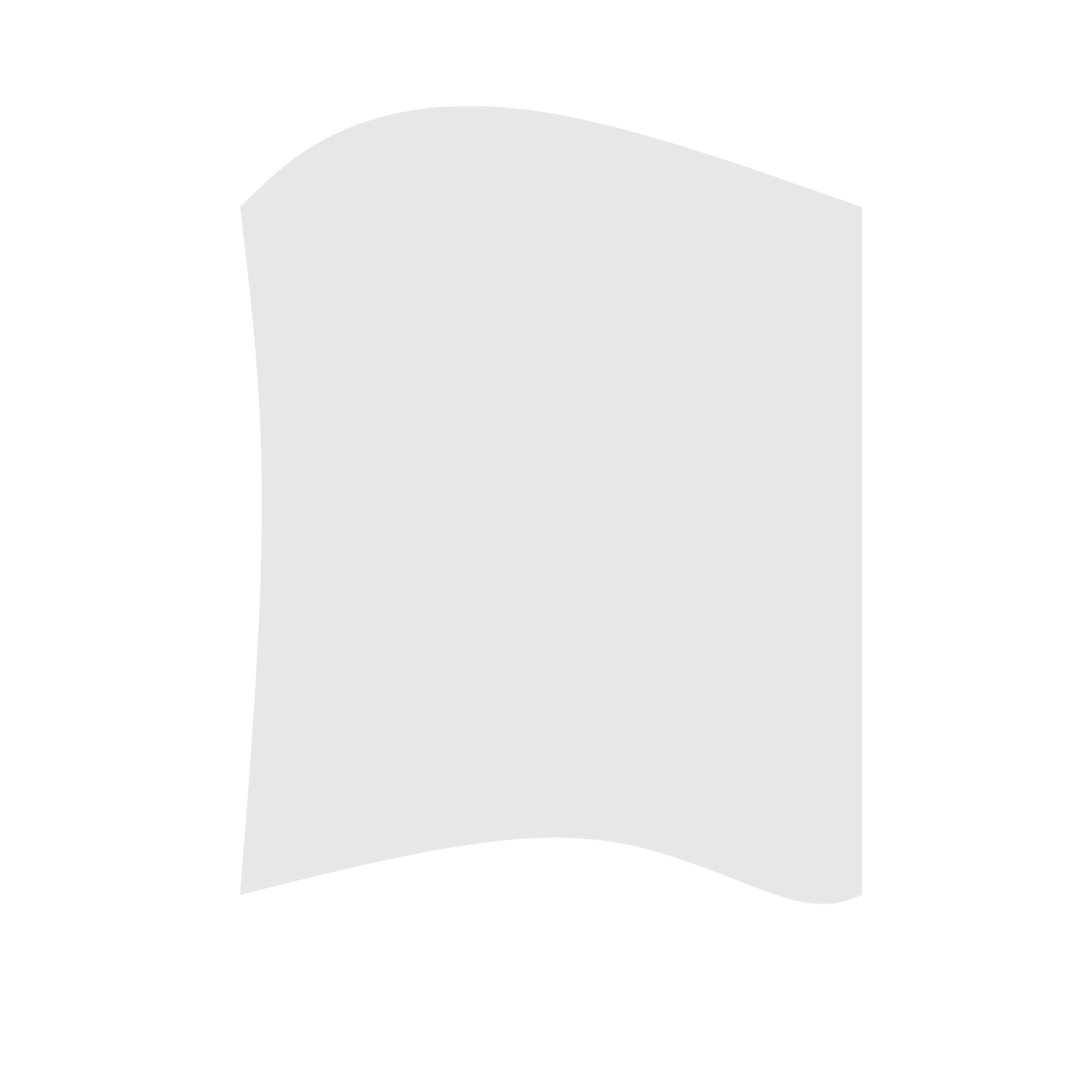
		\vspace{-.9cm}
		\caption{The Penrose diagram depicts the situation of a horizon. The endpoint of $\II$, denoted by $i^+$, is either in $\II$ or its completion; the latter case corresponds to future complete $\II$, the former one to an incomplete $\II$. We will show the absence of a horizon, that is the impossibility of either of those two cases.\label{fig:global_geometry}}
	\end{center}
	\end{minipage}
	\begin{minipage}{.1\textwidth}
	\end{minipage}
	\begin{minipage}{.45\textwidth}
	\begin{center}
		\vspace{-1.8cm}
		%\hspace{.2cm}
		\def\svgwidth{180pt}
		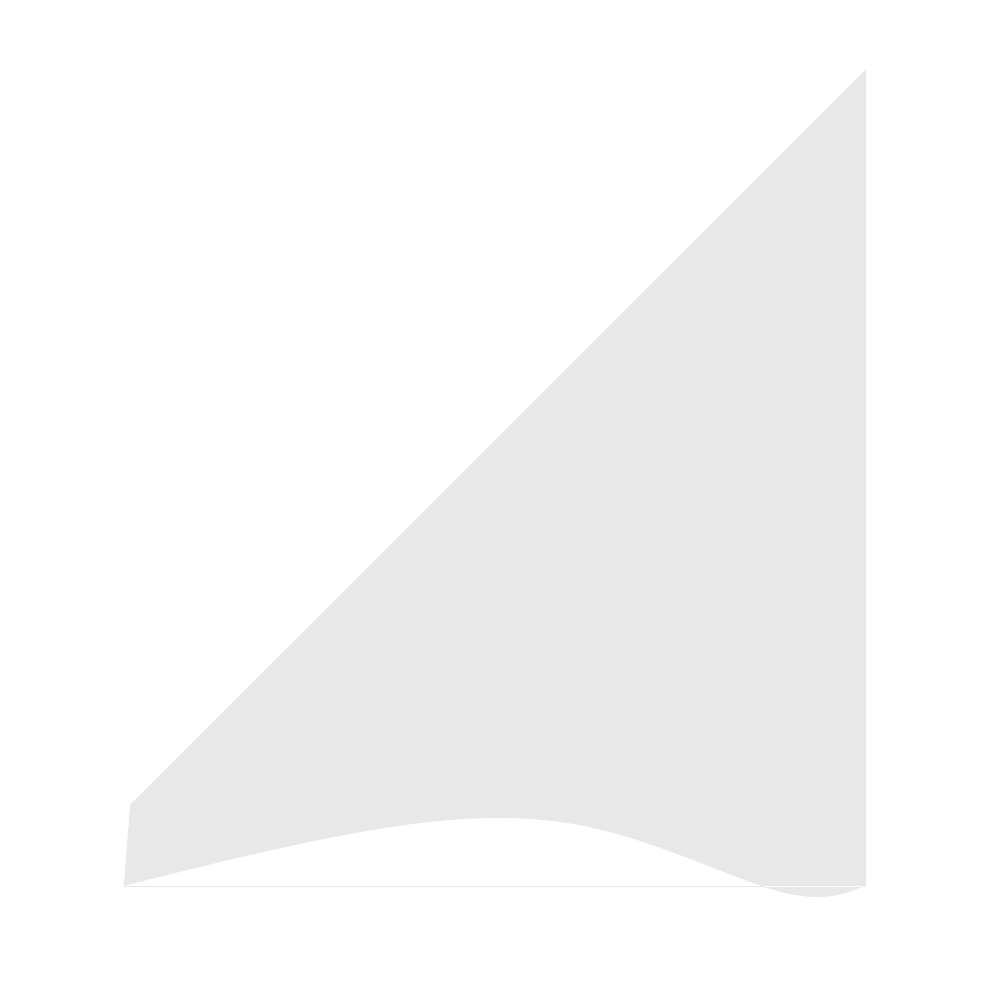
		\vspace{.67cm}
		\caption{The Penrose diagram depicts the formation of a first singularity emanating from the central wordline $\Gamma$.  The singularity is denoted by the circled dot; the dotted line is not part of the spacetime.\label{fig:naked_singularity}}
	\end{center}
	\end{minipage}
\end{figure}
%\end{wrapfigure}

%\begin{wrapfigure}{r}{.45\textwidth}
%\begin{figure}
%	\begin{center}
		%\vspace{.7cm}
		%\hspace{.2cm}
%		\def\svgwidth{250pt}
%		\input{naked_singularity.pdf_tex}
%		\caption{Formation of a naked singularity \label{fig:naked_singularity}}
%	\end{center}
%\end{figure}
%\end{wrapfigure}

In contrast, in a comparable context where no horizons can form, the work \citep{Bizon-three}, described in the next section, allows for growth of perturbations and global existence of the solution without the formation of a naked singularity. Thus the dynamics is very complicated and the question of the formation of naked singularities remains open.

\subsection{The significance of Eguchi-Hanson-AdS spacetimes}
\label{sec:physical_motivation}

The main motivation that sparked recent interest in asymptotically locally AdS solutions to the Einstein vacuum equations within the physics community is a putative connection between spacetimes of this form and conformal field theories defined on their respective boundaries: the AdS-CFT correspondence. It is of interest to understand what the positivity of gravitational energy means in the conformal field theory and thus `ground states', lowest energy configurations classically allowed, deserve consideration -- see \citep{Galloway-Woolgar} for more details and references on the issue of gravitational energy in this context.

A ground state depends heavily on the topology at infinity. If the spacetime is asymptotically AdS, this ground state is exact anti-de\,Sitter space with vanishing mass -- see \citep{Gibbons_uniqueness}. For asymptotically locally AdS spacetimes with toroidal topology at infinity, the works \citep{Horowitz-Myers}, \citep{Horowitz-Myers_positive} and \citep{Galloway-Woolgar} lend support to the conjecture that the so-called AdS soliton is the ground state in a suitable class of spacetimes.

The article \citep{Clarkson-Mann} was motivated by searching for a spacetime that asymptotically approaches $\mathrm{AdS}_5/\Gamma$, where $\Gamma$ is any freely acting discrete group of isometries, but has energy less than that of $\mathrm{AdS}_5/\Gamma$. This led to the Eguchi-Hanson-AdS solution in five dimensions.
These spacetimes have also been conjectured in \citep{Clarkson-Mann} to have minimal mass among asymptotically locally AdS spacetimes with topology $\mathrm{AdS}_5/\ZZ_n$ at infinity:

\begin{conj}
	Let $(\Ss,\ov g,K)$ be of Eguchi-Hanson type with $\Ss=(a,\infty)\times(S^3/\ZZ_n)$, then
	\begin{align*}
	M\geq M_{\mathrm{EH},a}
	\end{align*}
	with equality if and only if the data agree with those induced by the Eguchi-Hanson-AdS spacetime with parameter $a$.
\end{conj}

In a neighbourhood of Eguchi-Hanson-AdS solutions, this was indeed shown to be true:

\begin{thm}[\citep{Dafermos-Holzegel_EH}, see also \citep{Clarkson-Mann}]
	\label{thm:EH_minimising_mass}
	Given any $a>0$, assume initial data $(\Ss,\ov g,K)$ of Eguchi-Hanson type with $\Ss=(a,\infty)\times (S^3/\ZZ_n)$ which are a sufficiently small, but non-zero perturbation of the data induced by the Eguchi-Hanson-AdS spacetime with parameter $a$, then the mass $M$ at infinity satisfies
	\begin{align}
	\label{eqn:mass_gap}
	M_{\mathrm{EH},a}<M<0.
	\end{align}
\end{thm}

Motivated by the static uniqueness theorem for exact AdS space, one conjectures -- see \citep{Dafermos-Holzegel_EH}:

\begin{conj}
	\label{conj:mass_gap}
	There are no static, globally regular asymptotically locally AdS solutions to (\ref{eqn:EVE_five}) with topology $S^3/\ZZ_n$ with mass $M$ satisfying (\ref{eqn:mass_gap}).
\end{conj}

Thus, fixing $a$, the Eguchi-Hanson-AdS spacetime satisfying (\ref{eqn:mass_gap}) can be seen as the ground state in the biaxial Bianchi IX symmetry class. There is a folklore statement that such ground states would be stable under gravitational perturbations. However, in contrast, the present work, paired with the above conjectures, heuristically hints at an instability: Perturbing an Eguchi-Hanson-AdS spacetime slightly increases its mass at infinity, whilst remaining negative; therefore, by Corollary~\ref{cor:no_horizons_perturbations}, the future maximal development cannot contain a black hole, but by Conjecture~\ref{conj:mass_gap}, there is no static end state for the perturbation, which intimates that a first singularity forms, emanating from the centre. Therefore, such perturbations are potential candidates for examples of the formation of naked singularities.

It is interesting to note that a dual situation is found for perturbations of $\mathrm{AdS}_3$, as investigated in \citep{Bizon-three}. There, small perturbations of three-dimensional AdS space were studied numerically as solutions to the Einstein-scalar field system. The parallel to our case is that in three dimensions, there exists a mass threshold below which no black holes can form. In contrast, while the numerical computations of \citep{Bizon-three} suggest turbulence which cannot be terminated by a black hole formation, they provide evidence that small perturbations remain globally regular in time since the turbulence is too weak.

Finally, studying five-dimensional static spacetimes for various values of $\Lambda$ or, more precisely, classifying their four-dimensional Riemannian counterparts is still an active field of research in geometry. It is known that there are exactly four complete non-singular four-dimensional Ricci flat Riemannian spaces: Euclidean space, Eguchi-Hanson space, self-dual Taub-NUT space and Taub-Bolt space. See \citep{Gibbons_talk} for further details. Moreover, Eguchi-Hanson space has been used in geometric gluing constructions; see \citep{Biquard} and \citep{Brendle}. For more results in this realm, both classical and recent, see \citep{Belinskii-Gibbons-Page-Pope}, \citep{LeBrun}, \citep{EH_notes}, \citep{Brendle}
and references therein.

%\begin{wrapfigure}{r}{.55\textwidth}
\begin{figure}
		\begin{center}
			\vspace{-.9cm}
			%\hspace{2cm}
			\def\svgwidth{250pt}
			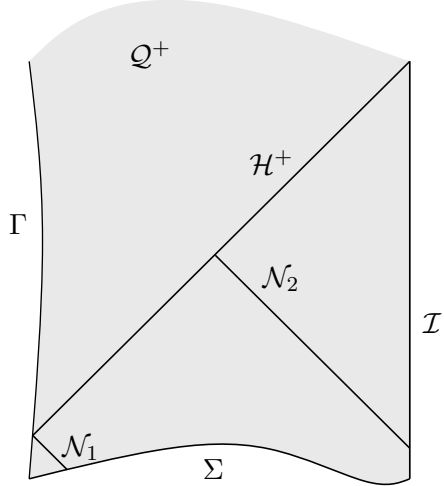
			\vspace{-.7cm}
			\caption{We can achieve that the initial data slice touches null infinity and does not reach $\Gamma$ by moving from a slice such as $\Nn_1$ to $\Nn_2$. \label{fig:globabl_geometry_moving_hypersurface}}
		\end{center}
\end{figure}
%\end{wrapfigure} 

\subsection{Outline of the paper}
\label{sec:outline}

From the local well-posedness theorem (Theorem~\ref{thm:general_local_existence}), we obtain the existence of a maximal development of Eguchi-Hanson-type data, with $B$ satisfying a Dirichlet boundary condition at infinity.  The global geometry of spacetimes arising from such data is described in Section~\ref{sec:global_biaxial_Bianchi_IX}. We also prove in that section that the spacetime is either globally regular without a horizon, or forms a horizon, or evolves into a first singularity at the centre. We proceed to show that no horizons can form in the dynamical evolution.

Proving the absence of horizons will take the structure of an argument by contradiction (Section~\ref{sec:a_priori_no_horizon}). Suppose that the Penrose diagram of the spacetime looks like Figure~\ref{fig:globabl_geometry_moving_hypersurface}. By soft arguments relying on the well-posedness result, we show that one can always find a null hypersurface such as $\Nn_2$ that does not intersect the initial hypersurface, but reaches from the horizon to null infinity. In Section~\ref{sec:extension_principle}, an extension principle is shown for triangular regions around null infinity -- such as the one enclosed by $\HH^+$, $\Nn_2$ and $\II$ --, which permits to extend the solutions to the future to a strictly larger triangle, provided uniform bounds hold in the triangular region. The proof of the extension principle uses the local existence result proved in \citep{Dold_thesis} and reviewed in Section~\ref{sec:local_well}. We then proceed (Section~\ref{sec:a_priori_no_horizon}) to establish that those quantities can indeed be bounded only in terms of their values on $\Nn_2$, where they hold by compactness of $\Nn_2$ and the local well-posedness result. Thus we can extend the solution along $\II$ beyond $\HH^+$, which is a contradiction.

\section{Local well-posedness}
\label{sec:local_well}

Before beginning with the proof of the main result, we review the local well-posedness result around null infinity of \citep{Dold_thesis} since the extension principle of Section~\ref{sec:extension_principle} relies on it. 
The exposition of this section parallels that of \citep{Holzegel-Warnick_Einstein-Klein-Gordon}. This will allow the reader familiar with that argument to gain quick access to the problem at hand.

Proving local well-posedness in the context of negative cosmological constant around infinity has been achieved for the four-dimensional Einstein-Klein-Gordon system in \citep{Holzegel-Smulevici_self-gravitating} and \citep{Holzegel-Warnick_Einstein-Klein-Gordon}. Several differences arise in the present context, which are outlined in \citep{Dold_thesis}. However, we can follow the general strategy of \citep{Holzegel-Warnick_Einstein-Klein-Gordon}. We first define 
the triangle
\begin{align*}
\Delta_{\delta,u_0}:=\{(u,v)\in\RR^2\,:\,u_0\leq v\leq u_0+\delta,\,v<u\leq u_0+\delta\}
\end{align*}
and the conformal boundary
\begin{align*}
\II:=\ov{\Delta_{\delta,u_0}}\backslash\Delta_{\delta,u_0}.
\end{align*}
See Figure~\ref{fig:triangular_domain} for a visualisation.
Our dynamical variables are
\begin{align*}
(\rrr,m,B):\,\Delta_{\delta,u_0}\rightarrow\RR_+\times\RR^2.
\end{align*}
We treat these as defining auxiliary variables
\begin{align}
\label{eqn:auxiliary}
r:=&\frac{1}{\rrr},~~~~~~~~1-\mu:=1-\frac{2m}{r^2}+\frac{r^2}{\ell^2},~~~~~~~~
\Omega^2:=-\frac{4r^4\rrr_u\rrr_v}{1-\mu}.
\end{align}

The general discussion of Appendix~\ref{sec:EVE_biaxial_Bianchi_IX_general} show that the correct notion of a solution to the Einstein vacuum equations in the triangular region is encapsulated in the following

\begin{wrapfigure}{r}{.3\textwidth}
	\vspace{.1cm}
	\hspace{.2cm}
	\def\svgwidth{200pt}
	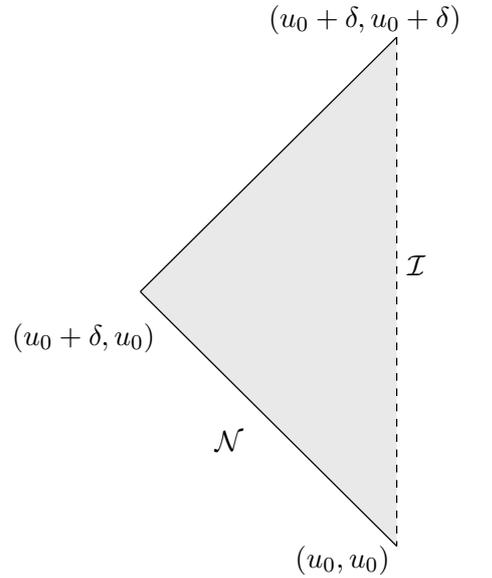
	\vspace{.7cm}
	\caption{The triangular domain $\Delta_{\delta',u_0}$ of local existence \label{fig:triangular_domain}}
\end{wrapfigure}

\begin{defn}
	\label{defn:defn_weak}
	A weak solution to the Einstein vacuum equations in $\Delta_{\delta,u_0}$ is a triple 
	\begin{align*}
	(\rrr,m,B)\in C^1_{\mathrm{loc}}(\Delta_{\delta,u_0})\cap W^{1,1}_{\mathrm{loc}}(\Delta_{\delta,u_0})\cap H^1_{\mathrm{loc}}(\Delta_{\delta,u_0})	
	\end{align*} 
	such that $\rrr_{uu},B_u,m_u\in C^0_{\mathrm{loc}}$ and the equations
	\begin{align}
	\label{eqn:rtilde_problem}
	\rrr_{uv}=&\Omega^2\rrr^3\left(-1+\frac{1}{3}R-2m\rrr^2\right)\\
	\label{eqn:mu_problem}
	\del_um=&-\frac{\rrr_u}{\rrr^3}\left(1-\frac{2}{3}R\right)+\frac{4}{\Omega^2}\frac{\rrr_v}{\rrr^5}\left(B_u\right)^2\\
	\label{eqn:mv_problem}
	\del_vm=&-\frac{\rrr_v}{\rrr^3}\left(1-\frac{2}{3}R\right)+\frac{4}{\Omega^2}\frac{\rrr_u}{\rrr^5}\left(B_v\right)^2\\
	\label{eqn:B_problem}
	B_{uv}=&\frac{3}{2}\frac{\rrr_u}{\rrr}B_v+\frac{3}{2}\frac{\rrr_v}{\rrr}B_u-\frac{1}{3}\Omega^2\rrr^2\left(\e^{-2B}-\e^{-8B}\right)
	\end{align}
	are satisfied in the interior of $\Delta_{\delta,u_0}$ in a weak sense.
\end{defn}

Equations (\ref{eqn:rtilde_problem}), (\ref{eqn:B_problem}) and (\ref{eqn:mv_problem}) are treated as the dynamical equations, whereas we will see that (\ref{eqn:mu_problem}) can be treated as a constraint equation that is propagated.

\begin{defn}
	\label{defn:initial_data}
	Let $\Nn=(u_0,u_1]$. A triple $(\ovt r,M,\ov B)\in C^2(\Nn)\times\RR\times C^1(\Nn)$ is a free data set if the following hold:
	\begin{compactenum}[(i)]
		\item $\ovt r>0$ and $\ovt r_u>0$ in $\Nn$, as well as $\lim_{u\rightarrow u_0}\ovt r(u)=0$, $\lim_{u\rightarrow u_0}\ovt r_u=1/2$ and $\lim_{u\rightarrow u_0}\ovt r_{uu}=0$.
		\item There is a constant $C_0$ such that
		\begin{align*}
		\int_{u_0}^{u_1}\frac{1}{(u-u_0)^3}\left[\ov B^2+(\ov{B}_u)^2\right]\,\dd u&<C_0\\
		\sup_{\Nn}\left\lvert \ovt r^{-2}\ov B\right\lvert+\sup_{\Nn}\left\lvert\ovt r^{-1}\del_u\ov{B}\right\lvert&<C_0.
		\end{align*}
	\end{compactenum}
\end{defn}

From this, we obtain a complete initial data set $(\ovt r,\ov B,\ov m,\ovt{r}_v)\in C^2(\Nn)\times C^1(\Nn)\times C^1(\Nn)\times C^1(\Nn)$. First, we integrate
\begin{align}
\label{eqn:ov_m}
\del_u\ov{m}=-\frac{\ovt{r}_u}{\ovt{r}^3}\left(1-\frac{2}{3}\ov{R}\right)-\frac{1}{\ovt{r}\ovt{r}_u}\left(1-2\ov{m}\ovt{r}^2+\frac{1}{\ell^2\ovt{r}^2}\right)(\ov{B}_u)^2
\end{align}
with boundary condition
\begin{align*}
\lim_{u\rightarrow u_0}\ov{m}=M
\end{align*}
for the Hawking mass $M$ at infinity.
The function $\ov{\rrr_v}$ is obtained from solving the ODE
\begin{align*}
\frac{\del_u\ov{\tilde{r}_v}}{\ov{\rrr_v}}=-\frac{4\ov{\rrr}_u}{\ovt{r}-2\ov{m}\ovt{r}^3+\frac{1}{\ell^2\ovt{r}}}\left(-1+\frac{\ov R}{3}-2\ov{m}\ovt{r}^2\right)
\end{align*}
with boundary condition
\begin{align*}
\lim_{u\rightarrow u_0}\ov{\rrr_v}=-\frac{1}{2}.
\end{align*}

\begin{thm}
	\label{thm:well_posed_main}
	Let $(\ovt r,M,\ov B)$ be a free data set in the sense of Definition~\ref{defn:initial_data} on $\Nn=(u_0,u_1]$. Then there is a $\delta>0$ such that there exists a unique weak solution $(\rrr,m,B)$ of the Einstein equations in the triangle $\Delta_{\delta,u_0}$ such that
	\begin{compactenum}[(i)]
		\item $\rrr$ satisfies the boundary condition
		\begin{align*}
		\rrr\big\lvert_{\II}=0
		\end{align*}
		\item $B$ satisfies the boundary condition
		\begin{align*}
		B\big\lvert_{\II}=0
		\end{align*}
		in a weak sense
		\item $\rrr$ and $B$ agree with $\ovt r$ and $\ov B$ when restricted to $\Nn$.
	\end{compactenum}
\end{thm}

\begin{rk}
	Imposing higher regularity on the initial data, we obtain a classical solution to (\ref{eqn:EVE_five}) in $\Delta_{\delta,u_0}$; see \citep{Dold_thesis} for a precise statement.
\end{rk}

We conclude this section with a remark about the Hawking mass. The mass is a dynamical variable and does not have to be conserved at infinity a priori. However, a geometric version of conservation holds:

\begin{propn}
	\label{propn:conserved_mass}
	Let $(\rrr,m,B)$ be a classical solution. Set
	\begin{align*}
	\Tt:=&\frac{1}{\Omega^2}\left(r_v\del_u-r_u\del_v\right)\\
	\Rr:=&\frac{1}{\Omega^2}\left(r_v\del_u+r_u\del_v\right)
	\end{align*}
	Then $\Tt$ and $\Rr$ are invariant under a change of the $(u,v)$ coordinates that preserves the form of the metric and
	\begin{align*}
	\Tt m\bigg\lvert_{\II}=0.
	\end{align*}
\end{propn}

\section{The global problem}

This section is devoted to studying the global dynamics arising from Eguchi-Hanson-type initial data. The existence of a maximal development is guaranteed by Theorem~\ref{thm:general_local_existence} and Remark~\ref{rk:maximal_development}. In Section~\ref{sec:global_biaxial_Bianchi_IX}, we specify our choice of coordinates on the orbits of the $SU(2)\times U(1)$ action and derive some geometric properties; here, we follow the exposition of \citep{Dafermos_naked_Higgs} and \citep{Dafermos_spherical_trapped} mutatis mutandis. Proving that the existence of a horizon would be contradictory is the content of Sections~\ref{sec:extension_principle} and \ref{sec:a_priori_no_horizon}.

\subsection{Global biaxial Bianchi IX symmetry}
\label{sec:global_biaxial_Bianchi_IX}

Let $(\Ss,\ov{g},K)$ be of Eguchi-Hanson type with negative mass $M$ at infinity. Then there is a unique maximal forward development $(\MM^+,g)$ by Theorem~\ref{thm:general_local_existence} and Remark~\ref{rk:maximal_development} which is asymptotically locally AdS. There is a projection map $\pi:\,\MM^+\rightarrow\Qq^+$ onto a two-dimnsional manifold with boundary $\Qq^+$ such that every $q\in\Qq^+$ represents an orbit under the $SU(2)\times U(1)$ symmetry. The manifold $\Qq^+$ can be embedded smoothly into $(\RR^2,g_{\mathrm{Mink}})$ and its boundary consists of a one-dimensional curve $\Sigma$ (initial hypersurface) and a one-dimensional curve $\Gamma$ (central worldline, where $r=0$). 
Choosing standard null coordinates $(u,v)$ on $\RR^{1+1}$, $\Qq^+$ shall be endowed with a metric
\begin{align*}
h=-\frac{1}{2}\Omega^2(u,v)\left(\dd u\otimes\dd v+\dd v\otimes\dd u\right).
\end{align*}
We choose $u$ such that the curves of constant $u$ are outgoing and such that $u$ as well as $v$ are increasing to the future along $\Gamma$. A coordinate chart $(u',v')$ preserves these assumptions if and only if
\begin{align}
\label{eqn:coordiante_change_condition}
\frac{\del u'}{\del u}>0,~\frac{\del v'}{\del v}>0,~\frac{\del u'}{\del v}=\frac{\del v'}{\del u}=0.
\end{align}
With respect to $h$, $\Sigma$ is spacelike and $\Gamma$ timelike.
Conformal infinity $\II\subseteq\ov{\Qq^+}\backslash\Qq^+$ is defined as follows: Set
\begin{align*}
\UU:=\left\{u\,:\,\sup_{(u,v)\in\Qq^+}r(u,v)=\infty\right\}.
\end{align*}
For each $u\in\UU$, there is a unique $v^{\ast}(u)$ such that
\begin{align*}
\left(u,v^{\ast}(u)\right)\in\ov{\Qq^+}\backslash\Qq^+.
\end{align*}
Note here that the closure is always taken with respect to the topology of $\RR^2$.
Now define null infinity as
\begin{align*}
\II:=\bigcup_{u\in\UU}(u,v^{\ast}(u)).
\end{align*}
Since the spacetime is asymptotically locally AdS, null infinity $\II$ is timelike. We have
\begin{align*}
\Qq^+=D^+\left(\Sigma\cup\II\right),
\end{align*}
i.\,e. $\Qq^+$ is in the future domain of dependence of $\Sigma$ and $\II$. By a simple change of coordinates satisfying (\ref{eqn:coordiante_change_condition}), we achieve that $u=v$ on $\II$. Note that in general, we cannot achieve that both $\II$ and $\Gamma$ are straightened out in this way.
We know that $B$ extends continuously to $\II$ and vanishes there. Moreover, we know that the Hawking mass $m$ extends continuously to $\II$ and equals a constant value $M<0$.

\begin{lemma}
		\label{defn:Sigma_admissible}
		The following hold:
	\begin{compactenum}[(i)]
		\item $r$ is unbounded on $\Sigma$.
		\item $m<0$ on $\Sigma$ and $m\rightarrow M<0$ as $r\rightarrow\infty$.
		\item $r_v>0$ at one point in $\Sigma$ and
		\begin{align*}
		\frac{r_v}{\Omega^2}\rightarrow c_0>0
		\end{align*}
		as $r\rightarrow\infty$.
	\end{compactenum}
\end{lemma}

\begin{proof}
	By the monotonicity of the mass and $M<0$, we immediately obtain that $m<0$ on $\Sigma$. The radius $r$ is unbounded on $\Sigma$ by the definition of Eguchi-Hanson-type data. For the Hawking mass to be finite at infinity,
	\begin{align*}
	-4\frac{r_ur_v}{\Omega^2}=-\frac{r^2}{\ell^2}+\OO(r)
	\end{align*}
	as $r\rightarrow\infty$. Moreover, by the above choice of $(u,v)$
	\begin{align*}
	\frac{r_u}{r_v}\rightarrow -1
	\end{align*}
	as $r\rightarrow\infty$. The conformal factor $\Omega^2$ grows as $r^2$ since the spacetime is asymptotically locally AdS. Therefore, we deduce that $r_v/\Omega^2$ is positive and finite as $r\rightarrow\infty$.
\end{proof}

\begin{wrapfigure}{r}{.35\textwidth}
%\begin{figure}
	\begin{center}
		\vspace{-.7cm}
		%\hspace{.2cm}
		\def\svgwidth{180pt}
		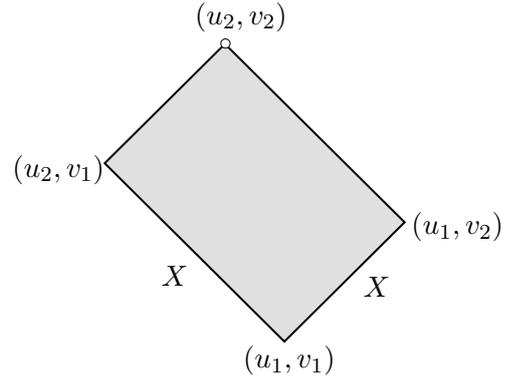
		\vspace{-.7cm}
		\caption{\footnotesize The extension property in the interior \label{fig:extension_double_null}}
	\end{center}
%\end{figure}
\end{wrapfigure}

\begin{propn}
	\label{propn:no_trapping}
	The above manifold $\MM^+$ does not have any trapped or marginally trapped surfaces, i.\,e.
	\begin{align*}
	r_{u}<0~\mathrm{and~} r_v>0
	\end{align*}
	globally in $\Qq^+$.
\end{propn}

\begin{proof}
	From
	\begin{align*}
	0>m=\frac{r^2}{2}\left(1+4\frac{r_ur_v}{\Omega^2}\right)+\frac{r^4}{2\ell^2}
	\end{align*}
	we conclude that $r_ur_v<0$ wherever $r$ is finite. Since there is a point on $\Sigma$ where $r_v>0$, the conclusion follows.
\end{proof}

\begin{rk}
	The absence of anti-trapped surfaces can also be guaranteed by fixing the sign of $r_u$ on $\Sigma$ and $\II$ and then using (\ref{eqn:constraint_u}).
\end{rk}

This fact already allows us to prove a weak geometric statement about the potential singularities that can arise in the time evolution.

\begin{defn}
	Let $p\in\ov{\Qq^+}$. The indecomposable past subset $J^-(p)\cap\Qq^+$ is said to be eventually compactly generated if there exists a compact subset $X\subseteq\Qq^+$ such that
	\begin{align*}
	J^-(p)\subseteq D^+(X)\cup J^-(X).
	\end{align*}
\end{defn}

Here we denote by $J^-(S)$ the causal past of a subset $S$.

\begin{defn}
	\label{defn:first_singularity}
	A point $p\in\ov{\Qq^+}\backslash\Qq^+$ is a first singularity if $J^-(p)\cap\Qq^+$ is eventually compactly generated and if any eventually compactly generated indecomposable past proper subset of $J^-(p)\cap\Qq^+$ is of the form $J^-(q)$ for a $q\in\Qq^+$.
\end{defn}

\begin{lemma}
	\label{thm:no_first_singularity}
	Let $p\in\ov{\Qq^+}\backslash\Qq^+$ be a first singularity. Then
	\begin{align*}
	p\in\ov{\Gamma}\backslash\Gamma.
	\end{align*}
\end{lemma}

\begin{proof}
	Suppose $p\notin\ov{\Gamma}$. Since the compact set $X$ of Definition~\ref{defn:first_singularity} has to be wholly contained in $\Qq^+$, we know that $p\notin\II$. In particular, $p$ is the future endpoint of a rectangle, whose remainder is completely contained in the interior of $\Qq^+$; see Figure~\ref{fig:extension_double_null}. By Proposition~\ref{propn:no_trapping}, $r_u<0$ and $r_v>0$ in this rectangle. Therefore, we can apply the standard extension principle away from infinity and the central worldline -- in a manner as e.\,g. in \citep{Dafermos_naked_Higgs} -- to conclude that $p\in\Qq^+$, a contradiction.
\end{proof}

\begin{thm}
	\label{thm:character_horizon}
	If $\Qq^+\backslash J^-(\II)\neq\emptyset$, then there is a null curve $\HH^+\subseteq\Qq^+$ such that
	\begin{align}
	\label{eqn:character_horizon}
	\HH^+=\ov{J^-(\II)}\backslash\left(I^-(\II)\cup\ov{\II}\right).
	\end{align}
\end{thm}

Note that $I^-(S)$ denotes the chronological past of a subset $S$.

\begin{proof}
	The horizon is given by
	\begin{align*}
	\HH^+=\mathring{J}^-(\II)\cap\Qq^+.
	\end{align*}
	Let $p$ be the future endpoint of the horizon. Since  $\Qq^+\backslash J^-(\II)\neq\emptyset$, $p\notin\ov{\Gamma}$. If $p\in\Gamma$, there is an open set $U$ such that $U\cap\Qq^+\neq\emptyset$ If $p\notin\ov{\Gamma}\cup\ov{\II}$, then $p$ is a first singularity and we have a contradiction to Lemma~\ref{thm:no_first_singularity}. Therefore $p\in\ov{\II}$.
\end{proof}

Sections~\ref{sec:extension_principle} and \ref{sec:a_priori_no_horizon}  are devoted to showing Theorem~\ref{thm:no_horizons_intro}.

\subsection{An extension principle}
\label{sec:extension_principle}

We formulate and prove an extension principle tailored to extending a solution beyond a supposed horizon.  

\begin{thm}
	\label{thm:extension_principle}
	Let $(\rrr,m,B)$ be a classical solution to the Einstein equations (\ref{eqn:rtilde_problem}), (\ref{eqn:B_problem}), (\ref{eqn:mu_problem}) and (\ref{eqn:mv_problem}) in the punctured triangle $\Delta:=\Delta_{d,u_0}\backslash\{(u_0+d,u_0+d)\}$. Let $\IIa=\ov{\Delta_{d,u_0}}\backslash\{\Delta_{d,u_0}\cup\{(u_0+d,u_0+d)\}\}$. Assume that
	\begin{align*}
	\rrr\Big\lvert_{\IIa}=0,~m\Big\lvert_{\IIa}=M<0,~B\Big\lvert_{\IIa}=0
	\end{align*}
	and that
	\begin{align*}
	\lim_{v\rightarrow u_0+d} \rrr(u_0+d,v)=0.
	\end{align*}
	Suppose that 
	\begin{align*}
	\rrr_{u}>0~\mathrm{and~} \rrr_v<0
	\end{align*}
	in $\Delta$, that 
	\begin{align*}
	\inf_{\IIa}\,\rrr_u>0 ~\mathrm{and~}\sup_{\IIa}\,\rrr_u<0
	\end{align*}
	and that there is a $C>0$ such that
	\begin{align*}
	\sup_{u_0\leq v\leq u_0+d}\int_{v}^{u_0+d}\rrr^{-3}\left(B^2+(B_u)^2\right)\,\dd u'+\sup_{\Delta}\left\lvert\rrr^{-1}\del_u B\right\lvert<C.
	\end{align*}
	Then there is a $\delta>0$ such that the solution $(\rrr,m,B)$ can be extended to the strictly larger triangle $\Delta_{d+\delta,u_0}$.
\end{thm}

\begin{proof}
	Extending beyond the domain of existence means using the local well-posedness result to extend the solution further into the future. We will first need to make sure that on each constant $v$-slice, the function $\rrr$ satisfies the correct boundary conditions.
	Reformulating equation (\ref{eqn:rtilde_problem}), we obtain
	\begin{align*}
	\rrr_{uv}=-\frac{4\rrr_u\rrr_v}{\rrr^2+2|m|\rrr^4+\ell^{-2}}\rrr\left(-1+\frac{1}{3}R-2m\rrr^2\right).
	\end{align*}
	Therefore $\rrr_{uv}=0$ on $\IIa$. Using a coordinate change, we want to fix the value of $\rrr_u$.	
	By the assumptions and since $\rrr_u=-\rrr_v$ on $\IIa$, $|\rrr_u|,|\rrr_v|\geq c>0$ and  
	\begin{align*}
	\frac{\dd u'}{\dd u}(u,v)=2\rrr_u(u,u),~\frac{\dd v'}{\dd v}=2\rrr_v(v,v),
	\end{align*}
	where $u'=u$ and $v'=v$ on $\II$,
	defines a regular change of coordinates that preserves the biaxial Bianchi IX symmetry. Moreover, in $(u',v')$ coordinates,
	\begin{align}
	\label{eqn:standard_boundary_conditions}
	\rrr\Big\lvert_{\II}=0,~\rrr_{u'}\Big\lvert_{\II}=\frac{1}{2},~\rrr_{u'v'}\Big\lvert_{\II}=0.
	\end{align}
	Hence we can assume without loss of generality that (\ref{eqn:standard_boundary_conditions}) already holds in the original $(u,v)$ coordinates.
	
	To increase the domain of existence, we also need initial $v$-slices of increased length. This can be achieved by an application of the standard local existence theorem away from infinity in double null coordinates whose proof proceeds by the same methods as for $\Lambda=0$, which is standard by now. Prescribing data on the slice $(u_0,u_0+d+\delta']$ of constant $v=u_0$ (for a $\delta'>0$), we infer that for every $\eps>0$, there is a $\delta'>0$ such that the solution can be extended to the set
	\begin{align*}
	\Delta_{\eps}:=\Delta\cup\left(\Delta_{d+\delta',u_0}\cap\{v\leq u_0+d-\epsilon\}\right).
	\end{align*}
	For each constant $v$-ray in $\Delta_{\eps}$, we have a initial data set, whose functions have norms uniformly bounded by $2C$. Note that the condition $1-\mu>r^2/\ell^2$ holds everywhere because $M<0$. 
	
	Therefore, there is a $\delta^{\ast}$ independent of $\eps$ such that each slice of constant $v$ in $\Delta_{\eps}$ yields a solution in a triangular domain of size $\delta^{\ast}$. Now we choose $\eps=\delta^{\ast}/2$ and see that the solution $(\rrr,m,B)$ extends to a strictly larger triangle $\Delta_{d+\delta,u_0}$, where $\delta=\eps$.
\end{proof}

The proof above yields a another version of the extension principle that we formulate separately for the sake of clarity.

\begin{cor}
	\label{cor:extension_principle}
	Suppose the assumptions of Theorem~\ref{thm:extension_principle} hold. Moreover, let us assume that the classical solution on $\Delta$ has an extension to the extended initial data slice $\tilde\Nn=(u_0,u_0+d+\eps]$.  
	Then there is a $\delta>0$ such that the solution $(\rrr,m,B)$ can be extended to $\Delta_{d+\delta,u_0}$ such that it agree on $\tilde{\Nn}\cup \Delta_{d+u_0,u_0}$ with the given values. Furthermore, the extension is unique for sufficiently small $\delta>0$ 
\end{cor}

\subsection{A priori estimates}
\label{sec:a_priori_no_horizon}

In this section, we first establish what was described through Figure~\ref{fig:globabl_geometry_moving_hypersurface} in Section~\ref{sec:outline} as the soft argument. This is the content of Lemma~\ref{lemma:start_a_priori}. The remainder of the section contains the argument by contradiction, using the extension principle in form of Corollary~\ref{cor:extension_principle}.

\begin{lemma}
	\label{lemma:start_a_priori}
	Let $\Qq^+$ be as in Theorem~\ref{thm:no_horizons_intro}. Set $\Delta^d_{u}:=\Delta_{d,u}\backslash\{(u+d,u+d)\}$ and $\Nn_u^d:=\{v=u\}\cap\Delta^d_u$. Then for any $\Delta^{d_1}_{u_1}\subseteq\Qq^+$, there is a $u_0\geq u_1$ and $d_0:=d_1-(u_0-u_1)$ such that
	\begin{compactenum}[(i)]
		\item $r\geq r_0>0$ in $\Delta^{d_0}_{u_0}$.
		\item There are $q_1,q_2>0$ such that
		\begin{align}
		\label{eqn:upper_lower_bound_rv_Omega2}
		q_1\leq \frac{r_v}{\Omega^2}\leq q_2
		\end{align}
		in $\Delta_{u_0}^d$. The constants $q_1$ and $q_2$ depend on the choice of $r_0$.
	\end{compactenum}
\end{lemma}

\begin{proof}
	If $\Delta^{d_1}_{u_1}$ touches $\Gamma$, then by moving the initial slice of constant $v$ to the future -- as depicted in Figure~\ref{fig:globabl_geometry_moving_hypersurface} --, we achieve that $r\geq r_0>0$ since $r_v>0$ globally by Proposition~\ref{propn:no_trapping}.
	
	By assumption, the bound on $r_v/\Omega^2$ holds on $\Sigma$. Set
	\begin{align*}
	\Rr:=\{r=r_0\}\cap\{v\leq u_0\}\cap\Qq^+.
	\end{align*}
	The set $\Rr$ is closed and touches $\HH^+$ and $\Sigma$. The continuous function $r_v/\Omega^2$ is positive in $\Qq^+$ by Proposition~\ref{propn:no_trapping}. Therefore the bound on $r_v/\Omega^2$ holds in $\Rr$. We will show that (\ref{eqn:upper_lower_bound_rv_Omega2}) holds in the causal future of 
	\begin{align*}
	S:=\Rr\cup\left(\Sigma\cap\{r\geq r_0\}\right)
	\end{align*} 
	such that the constants $q_1$ and $q_2$ depend on the values of $r_v/\Omega^2$ on $S$. We rewrite the constraint equation (\ref{eqn:constraint_v}) as
	\begin{align}
	\label{eqn:bound_on_rv_Omega}
	\del_v\left(\frac{r_v}{\Omega^2}\right)=\left(-\frac{4r^3r_u}{\Omega^2}\left(B_v\right)^2\right)\left(-\frac{2}{r^2(1-\mu)}\right)\frac{r_v}{\Omega^2}.
	\end{align}
	Given $(u,v)\in J^+(S)$, there is a $(u',v')\in S$ such that $(u,v)$ and $(u',v')$ are connected by a null curve. We integrate (\ref{eqn:bound_on_rv_Omega}) along a ray of constant $u$ to find
	\begin{align*}
	\left(\frac{r_v}{\Omega^2}\right)(u,v)=&\exp\left(\int_{v'}^v\frac{4r^3r_u}{\Omega^2}\left(B_v\right)^2\frac{2}{r^2(1-\mu)}\,\dd v''\right)\cdot\left(\frac{r_v}{\Omega^2}\right)(u',v')\\
	\geq& \exp\left(-\frac{2}{r_0^2}\int_{u_0}^v\frac{-4r^3r_u}{\Omega^2}\left(B_v\right)^2\,\dd v''\right)\cdot\left(\frac{r_v}{\Omega^2}\right)(u',v')\\
	\geq& \exp\left(-\frac{2}{r_0^2}\left(m(u,v)-m(u',v')\right)\right)\cdot\left(\frac{r_v}{\Omega^2}\right)(u',v')
	\end{align*}
	For the first inequality, we have used $1-\mu>1$ and $r\geq r_0$. For the second inequality, we have used (\ref{eqn:mv_problem}) and have dropped a non-negative term.
	Therefore, we obtain
	\begin{align}
	\label{eqn:bound_rv_Omega2}
	\e^{-\frac{2}{r_0^2}\left[M-m(u_0+d,u_0)\right]}\frac{r_v}{\Omega^2}\bigg\lvert_{(u',v')}\leq \frac{r_v}{\Omega^2}\bigg\lvert_{(u,v)}\leq \frac{r_v}{\Omega^2}\bigg\lvert_{(u',v')}.
	\end{align}
	This yields (\ref{eqn:upper_lower_bound_rv_Omega2}).
\end{proof}

\begin{rk}
	A bound of the form \eqref{eqn:bound_rv_Omega2} can always be achieved, independently of the exact value of $M$.
\end{rk}

Now assume for the sake of contradiction that $\Qq^+$ possesses a horizon $\HH^+$. According to Lemma~\ref{lemma:start_a_priori}, we find a $\Delta:=\Delta_{u_0}^{d_0}$ such that $(u_0+d,u_0)\in\HH^+$ and such that the conclusions of the lemma hold. In particular, the constants and bounds will be fixed henceforth. Again, let $\IIa:=\ov{\Delta_{d,u_0}}\backslash\{\Delta_{d,u_0}\cup\{(u_0+d,u_0+d)\}\}$. Let $\Nn:=\Nn_{u_0}^{d_0}$. We always have that
\begin{align*}
\rrr\Big\lvert_{\II}=0,~m\Big\lvert_{\II}=M<0,~B\Big\lvert_{\II}=0,
\end{align*}
that
\begin{align*}
r_{u}<0~\mathrm{and~} r_v>0.
\end{align*}

We will show that all the assumptions of the extension principle hold. Let us first turn to estimating the norms of $B$.
The mass achieves its minimum at $(u_1,u_0)$ and its maximum on $\IIa$. Therefore, upon integration over constant $v$, we obtain
\begin{align*}
\int_{v}^{u_1}\left(-rr_u\left(1-\frac{2}{3}R\right)+\frac{4}{\Omega^2}r^3r_v(B_u)^2\right)\,\dd u=M-m(u_1,v)\leq M-m(u_1,u_0).
\end{align*}
Note that we have
	\begin{align*}
	1-\frac{2}{3}R\geq \min\{B^2/2,1\}.
	\end{align*}

We need to estimate the coefficients in the integral.
From
\begin{align*}
\rrr_u=\frac{1}{4r^2\frac{r_v}{\Omega^2}}(1-\mu)
\end{align*}
and (\ref{eqn:bound_on_rv_Omega}),
we obtain
\begin{align}
\label{eqn:inf_sup_ru}
\frac{1}{4\ell^2}\left(\max_{u_0\leq u\leq u_1}\frac{r_v}{\Omega^2}\bigg\lvert_{(u,u_0)}\right)^{-1}\leq \rrr_u\leq& \e^{\frac{2}{r_0^2}\left[M-m(u_0+d,u_0)\right]}\left(\min_{u_0\leq u\leq u_1}\frac{r_v}{\Omega^2}\bigg\lvert_{(u,u_0)}\right)^{-1}\times\\&~~~~~~~~~\times\frac{\rrr^2}{4}\left(1+\frac{2|M|}{r^2}+\frac{r^2}{\ell^2}\right)\notag
\end{align}
and see that $\rrr_u$ is uniformly bounded above and below by a constant depending only on data on $\Nn$. Therefore
\begin{align*}
C_1(u-v)\leq \rrr\leq C_2(u-v)
\end{align*}
and
\begin{align}
\label{eqn:limit_rrr}
\lim_{v\rightarrow u_0+d} \rrr(u_0+d,v)=0.
\end{align}
Furthermore,
\begin{align*}
-r_u\bigg\lvert_{(u,v)}=\frac{1+\frac{2|m|}{r^2}+\frac{r^2}{\ell^2}}{4\frac{r_v}{\Omega^2}}\geq \frac{r^2}{4\ell^2}\frac{1}{\max_{u_0\leq u\leq u_1}\frac{r_v}{\Omega^2}\bigg\lvert_{(u,u_0)}}.
\end{align*}
Therefore there is a constant $C_{u}$ depending only on values of $\rrr$, $\rrr_v$ and $m$ on $v=u_0$ such that
\begin{align}
\label{eqn:H1_u_unimproved}
\int_{v}^{u_1}r^3\left(\min\{B^2,1\}+(B_u)^2\right)\,\dd u<C_{u}
\end{align}
uniformly in $\Delta$.

Thus,
\begin{align*}
|B(u,v)|\leq \left(\int_v^u\frac{1}{r^3}\,\dd u'\right)^{1/2}\left(\int_v^ur^3B_u^2\,\dd u'\right)^{1/2}\leq \frac{C_2^{1/2}}{2}C_{u}^{1/2}(u-v)^{2}
\end{align*}
It follows that
\begin{align}
\label{eqn:bound_r-2B}
\rrr^{-2}|B|\leq C_{\mathrm{pointwise}}
\end{align}
uniformly.
Together with (\ref{eqn:H1_u_unimproved}), this yields
\begin{align}
\label{eqn:bound_H1u}
\int_{u}^{u_1}r^3\left(B^2+(B_u)^2\right)\,\dd u<C'_u.
\end{align}

In a similar way, one also obtains
\begin{align*}
\int_{u_0}^vr^3\left(B^2+(B_v)^2\right)\,\dd v<C'_v
\end{align*}
for $v\in[u_0,u_1)$ from integrating $\del_vm$ and then deriving a bound on $r_u/\Omega^2$ as (\ref{eqn:bound_rv_Omega2}). Here one uses that
\begin{align*}
\frac{r_u}{\Omega^2}\bigg\lvert_{\II}=-\frac{r_v}{\Omega^2}\bigg\lvert_{\II}.
\end{align*}

Using the wave equation for $B$ in the form
\begin{align}
\label{eqn:wave_transport_a_priori}
\del_v\left(r^{3/2}B_u\right)=-\frac{3}{2}r^{1/2}r_uB_v-\frac{\Omega^2}{3r^{1/2}}\left(\e^{-2B}-\e^{-8B}\right)
\end{align}
yields
\begin{align*}
\left\lvert r(u,v)^{3/2}B_u(u,v)\right\lvert\leq& r(u,u_0)^{3/2}|B_u(u,u_0)|+C\int_{u_0}^vr(u,v')^{5/2}|B_v(u,v')|\,\dd v'\\
&+\frac{1}{3}\left(\int_{u_0}^v\frac{\Omega^2}{r^2}\,\dd v'\right)^{1/2}\left(\int_{u_0}^v\Omega^2r\left\lvert\e^{-2B}-\e^{-8B}\right\lvert^2\,\dd v'\right)^{1/2}
\end{align*}
Since $\Omega^2/r^2$ is bounded by virtue of the bounds established above, the third term is easily seen to be bounded. The second term is estimates as
\begin{align*}
\int_{u_0}^vr^{5/2}|B_v|\,\dd v'\leq& \left(\int_{u_0}^v r^3(B_v)^2\,\dd v'\right)^{1/2}\left(\int_{u_0}^v r^2\dd v'\right)^{1/2}\\\leq& C\left(\int_{u_0}^v\frac{\del_v r}{(-\rrr_v)}\,\dd v'\right)^{1/2}\\\leq& Cr(u,v)^{1/2}.
\end{align*}
Therefore,
\begin{align*}
	|rB_u|\leq r(u,u_0)|B_u(u,u_0)|+C,
\end{align*}
where $C$ depends on values on $\Nn$ and on $C_1$, $C_2,$, $C'_v$. Since $B$ is a classical solution up to and including the horizon, there is an $\alpha>0$ such that
\begin{align}
\label{eqn:asymptotic_expansion}
B=\rho^{\alpha}\left(a_0(t)+a_1(t)\rho+o(\rho)\right)
\end{align}
for smoothly differentiable functions $a_0$ and $a_1$ of $t=(u+v)/2$. From above, the asymptotics of $r$, $r_u$, $r_v$ and $\Omega^2$ are known as we approach the boundary. Inserting \eqref{eqn:asymptotic_expansion} into \eqref{eqn:wave_transport_a_priori}, we obtain $\alpha=4$. Therefore, $rB_u$ is bounded on $\Nn$ and we have established the desired pointwise bound on $rB_u$ in $\Delta$.

An application of the extension principle (Theorem~\ref{thm:extension_principle}) yields Theorem~\ref{thm:no_horizons_intro} if it also holds true that
\begin{align*}
\inf_{\IIa}\,\rrr_u>0 ~\mathrm{and~}\sup_{\IIa}\,\rrr_u<0.
\end{align*}
This has been established already in (\ref{eqn:inf_sup_ru}), thus finishing the proof of Theorem~\ref{thm:no_horizons_intro}.

\section*{Acknowledgement}

I would like to thank Mihalis Dafermos and Gustav Holzegel for their guidance and support during the work on my PhD thesis, a part of which is the content of this paper. I gratefully acknowledge the financial support of EPSRC, the Cambridge Trust and the Studienstiftung des deutschen Volkes.

\begin{appendices}

	\section{Biaxial Bianchi IX symmetry}
	
	\subsection{Formulae related to the renormalised Hawking mass}
	
	In this section, we collect some useful calculations and identities related to the renormalised Hawking mass. Let us first start generally with an $n$-dimensional Lorentzian manifold $(\MM,g)$ with Levi-Civita connection $\nabla$ and a spacelike hypersurface $(\Nn,\ov{g})$ with induced second fundamental form $K$. Let $n$ be the timelike normal on $\Nn$. Let $(\Sigma,\gamma)$ be a compact three-dimensional submanifold of $\Nn$, separating $\Nn$ into an inside and an outside. Let $\nu$ be the unit normal pointing outside. Set $l_{\pm}:=n\pm\nu$. For $X$ and $Y$ tangent to $\Sigma$, we define the symmetric null second forms
	\begin{align*}
	\chi_{\pm}(X,Y):=g\left(\nabla_Xl_{\pm},Y\right).
	\end{align*}
	Clearly,
	\begin{align*}
	\chi_{\pm}=\frac{1}{2}\LL_{l_{\pm}}g.
	\end{align*}
	The associated null expansion scalars are defined by
	\begin{align*}
	\theta_{\pm}:=\spur_{\gamma}\chi_{\pm}=\gamma^{AB}\left(\chi_{\pm}\right)_{AB}=\gamma^{AB}\nabla_A\left(l_{\pm}\right)_B.
	\end{align*}
	Let $H$ be the mean curvature of $\Sigma$. Then we immediately obtain
	\begin{align*}
	\theta_{\pm}=\spur_{\gamma}K\pm H.
	\end{align*}
	
	Let us now assume that $\Nn$ is foliated by topological  3-spheres $\Sigma_r$, where
	\begin{align*}
	4\pi^2 r=\mathrm{Vol}\left(\Sigma_r\right)=\int_{\Sigma_r}\dd\mu_{\gamma}=\int_{S^3}\sqrt{\gamma}\,\dd\mu_{S^3}.
	\end{align*}
	We compute
	\begin{align*}
	\LL_{l_{\pm}}\left(\mathrm{Vol}\left(\Sigma_r\right)\right)=\int_{S^3}\frac{1}{2\sqrt{\gamma}}\LL_{l_{\pm}}\left(\det\gamma\right)\,\dd\mu_{S^3}=\int_{S^3}\frac{1}{2}\gamma^{AB}\left(\LL_{l_{\pm}}\gamma\right)_{AB}\sqrt{\gamma}\,\dd\mu_{S^3}=\int_{\Sigma_r}\theta_{\pm}\,\dd\mu_{\Sigma_r}.
	\end{align*}
	Assuming that $\theta_{\pm}$ is constant on $\Sigma_r$, we obtain
	\begin{align*}
	l_{\pm}r=\frac{r}{3}\spur_{\gamma}\chi_{\pm}.
	\end{align*}
	Therefore,
	\begin{align}
	\label{eqn:delr_hypersurface}
	g(\nabla r,\nabla r)=-\frac{r^2}{12}\left(\spur_{\gamma}\chi_+\right)\left(\spur_{\gamma}\chi_{-}\right)=-\frac{r^2}{12}\left(\left(\spur_{\gamma}K\right)^2-H^2\right).
	\end{align}
	
	\subsection{The Einstein vacuum equations reduced by biaxial Bianchi~IX symmetry}
	\label{sec:EVE_biaxial_Bianchi_IX_general}
	
	We quote the following result from \citep{Dafermos-Holzegel_EH} and do not present a derivation:
	
	\begin{thm}
		\label{thm:EVE_symmetry_full}
		Let $(\MM,g)$ exhibit a biaxial Bianchi IX symmetry with metric
		\begin{align*}
		g=-\frac{1}{2}\Omega^2\left(\dd u\otimes\dd v+\dd v\otimes\dd u\right)+\frac{1}{4}r^2\e^{2B}\left(\sigma_1^2+\sigma_2^2\right)+\frac{1}{4}r^2\e^{-4B}\,\sigma_3^2
		\end{align*}		
		Then the Einstein vacuum equations (\ref{eqn:EVE_five})
		for $\Lambda=-6/\ell^2<0$, understood as a classical system of partial differential equations, are equivalent to the system of two constraint equations
		\begin{align}
		\label{eqn:constraint_u}
		\del_u\left(\frac{r_u}{\Omega^2}\right)=&-\frac{2r}{\Omega^2}\left(B_u\right)^2\\
		\label{eqn:constraint_v}
		\del_v\left(\frac{r_v}{\Omega^2}\right)=&-\frac{2r}{\Omega^2}\left(B_v\right)^2
		\end{align}
		and four evolution equations
		\begin{align}
		\label{eqn:wave_r}
		r_{uv}=&-\frac{\Omega^2 R}{3r}-\frac{2r_ur_v}{r}-\frac{\Omega^2r}{\ell^2}\\
		\label{eqn:wave_omega}
		\left(\log\Omega\right)_{uv}=&\frac{\Omega^2 R}{2r^2}+\frac{3}{r^2}r_ur_v-3B_uB_v+\frac{\Omega^2}{2\ell^2}\\
		\label{eqn:wave_B}
		B_{uv}=&-\frac{3}{2}\frac{r_u}{r}B_v-\frac{3}{2}\frac{r_v}{r}B_u-\frac{\Omega^2}{3r^2}\left(\e^{-2B}-\e^{-8B}\right),
		\end{align}
		where
		\begin{align*}
		R=&2\e^{-2B}-\frac{1}{2}\e^{-8B}
		\end{align*}
		is the scalar curvature of the group orbits.
	\end{thm}
	
	\begin{proof}
		One needs to compute the components of the Ricci curvature. The constraints are the $uu$ and $vv$ components. The equation \eqref{eqn:wave_omega} comes from the $uv$ component. The equations \eqref{eqn:wave_r} and \eqref{eqn:wave_B} are the content of the remaining components.
	\end{proof}

	Now our aim is to reformulate the equations of Theorem~\ref{thm:EVE_symmetry_full} in terms of the renormalised Hawking mass.
	The derivatives of the mass are given by
	\begin{align}
	\label{eqn:m_u}
	\del_um=&rr_u\left(1-\frac{2}{3}R\right)-\frac{4}{\Omega^2}r^3r_v(B_u)^2\\
	\label{eqn:m_v}
	\del_vm=&rr_v\left(1-\frac{2}{3}R\right)-\frac{4}{\Omega^2}r^3r_u(B_v)^2.
	\end{align}
	Since $x\mapsto 2\e^{-2x}-\e^{-8x}/2$ is positive for $|x|\rightarrow\infty$ and has no extremum, we conclude that $\del_um\leq 0$ and $\del_vm\geq 0$.
	
	\begin{propn}
		\label{propn:weak_defn_EVE}
		Assume that (\ref{eqn:wave_r}), (\ref{eqn:wave_B}), (\ref{eqn:m_u}) and (\ref{eqn:m_v}) hold. Moreover, define $\Omega$ via (\ref{eqn:auxiliary}). Then the constraints (\ref{eqn:constraint_u}) and (\ref{eqn:constraint_v}) hold. If the right hand side of (\ref{eqn:wave_r}) can be differentiated in $u$, then also (\ref{eqn:wave_omega}) holds.
	\end{propn}
	
	\begin{proof}
		The proof is a calculation. We obtain
		\begin{align*}
		\frac{4r_u}{\Omega^2}=-\frac{1}{r_v}+\frac{2m}{r^2r_v}-\frac{r^2}{\ell^2r_v}
		\end{align*}
		from (\ref{eqn:auxiliary}). This yields
		\begin{align*}
		\del_u\left(\frac{r_u}{\Omega^2}\right)=&-\frac{r_u}{r_v}\frac{r_{uv}}{\Omega^2}+\frac{1}{2}\frac{m_u}{r^2r_v}-\frac{mr_u}{r^3r_v}-\frac{1}{2}\frac{r}{\ell^2}\frac{r_u}{r_v}\\
		=&\frac{r_u}{r_v}\left(\frac{1}{2r}-\frac{R}{3r}+\frac{R}{3r}-\frac{1}{2r}+\frac{m}{r^3}-\frac{r}{2\ell^2}+\frac{r}{\ell^2}-\frac{1}{2}\frac{r}{\ell^2}-\frac{m}{r^3}\right)-\frac{2}{\Omega^2}r\left(B_u\right)^2\\
		=&-\frac{2}{\Omega^2}r\left(B_u\right)^2.
		\end{align*}
		The second constraint equation is obtained analogously. To obtain the equation for $\left(\log\Omega\right)_{uv}$, we multiply (\ref{eqn:constraint_u}) by $\Omega^2$ and differentiate with respect to $u$:
		\begin{align*}
		-2r_u\left(\log\Omega\right)_{uv}+r_{uuv}-2r_{uv}\frac{\Omega_u}{\Omega}=-2r_v\left(B_u\right)^2-4rB_uB_{uv}
		\end{align*}
		Using (\ref{eqn:wave_r}), we obtain
		\begin{align*}
		r_{uuv}-2r_{uv}\frac{\Omega_u}{\Omega}=\frac{\Omega^2}{r^2}r_uR+\frac{4\Omega^2}{3r}B_u\left(\e^{-2B}-\e^{-8B}\right)+\frac{\Omega^2r_u}{\ell^2}+6\frac{(r_u)^2}{r^2}r_v+4r_v(B_u)^2.
		\end{align*}
		Applying (\ref{eqn:wave_B}), the desired equation follows.
	\end{proof}

\end{appendices}

\bibliographystyle{alphadin}

%\nocite{*}

\addcontentsline{toc}{section}{Bibliography}
{	\footnotesize \bibliography{literature_Eguchi-Hanson}}

\end{document}